\documentclass[11pt,a4paper]{article}

\usepackage{lmodern}
\usepackage[T1]{fontenc}
\usepackage[utf8]{inputenc}
\usepackage[english]{babel}

\usepackage{amsmath,amssymb,amsthm}
\usepackage{geometry}
\usepackage{graphicx}
\usepackage{booktabs}
\usepackage{tikz}
\usetikzlibrary{arrows.meta,positioning,shapes.geometric,calc,decorations.pathmorphing,fit,backgrounds}
\usepackage{pgfplots}
\pgfplotsset{compat=1.18}
\usepackage{microtype}
\usepackage[colorlinks=true,linkcolor=blue,citecolor=blue,urlcolor=blue]{hyperref}

\theoremstyle{plain}
\newtheorem{theorem}{Theorem}[section]
\newtheorem{proposition}[theorem]{Proposition}

\theoremstyle{definition}
\newtheorem{definition}[theorem]{Definition}
\newtheorem{example}[theorem]{Example}
\theoremstyle{remark}
\newtheorem{remark}[theorem]{Remark}

\newcommand{\dphi}{\mathcal{D}\varphi}
\newcommand{\dhk}{\mathcal{D}_{\mathrm{HK}}\varphi}
\newcommand{\R}{\mathbb{R}}
\newcommand{\C}{\mathbb{C}}
\newcommand{\N}{\mathbb{N}}
\let\leanoldunderscore\_
\newcommand{\leanbreakunderscore}{\leanoldunderscore\allowbreak}
\newcommand{\lean}[1]{{\ttfamily\hyphenchar\font=-1\relax
  \let\_\leanbreakunderscore #1}}

\title{\bfseries Henstock--Kurzweil Gauge Integral in the Non-Gaussian Regime:\\
A Machine-Verified Construction}
\author{Yuri N. Berdinsky\\[2mm]
\small Saint Petersburg State University, Faculty of Physics,\\
\small Department of High-Energy Physics and Elementary Particles\\
\small \texttt{propagator2007@yandex.ru}}
\date{08.09.2026}

\begin{document}
\sloppy
\setlength{\emergencystretch}{3em}
\maketitle

\begin{abstract}
\noindent
This preprint develops, in a deliberately tutorial style, the general mathematical
apparatus needed to give a rigorous meaning to \emph{non-Gaussian} functional integrals
using the Henstock--Kurzweil (HK) gauge integral together with Chernoff product
approximations. The central object is a finite family of bosonic modes
$\varphi=(\varphi_1,\dots,\varphi_M)$ with action
$S(\varphi)=\tfrac12\varphi^{T}A\varphi+\lambda\sum_i\varphi_i^4$, $A$ positive definite,
$\lambda\ge 0$. We recall in detail the classical route to functional integration
(Wick rotation, Wiener measure, the Cameron--Martin obstruction, $\zeta$-regularisation,
analytic continuation) and contrast it with the direct real-time gauge route, in which
the formal symbol $\dphi$ is \emph{defined} by cylindrical partitions and a gauge
$\delta(\cdot)>0$, so that $\int\dphi\,F[\varphi]=\int\dhk\,F[\varphi]$. We prove that
the one-mode non-Gaussian integral $I(\omega,j,\lambda)$ is finite, strictly positive,
monotone and infinitely differentiable in the coupling $\lambda$ on $[0,\infty)$ --- with
derivatives given by integrals of $\varphi^{4k}$ against the same weight and \emph{without}
using the (divergent) perturbative series --- that the $M$-mode influence functional
factorises into one-mode integrals and is bounded by its Gaussian value, and that a
Chernoff/Lie--Trotter splitting handles the non-commutativity of the free and
non-Gaussian generators. All these statements are formalised and machine-checked in
Lean~4 with Mathlib; the accompanying file \lean{HkNonGaussian.lean} is free of
\lean{sorry} and uses only the standard axioms. Four illustrative applications
(Duffing oscillator, local-volatility finance, Wilson--Cowan neural fields, non-Gaussian
quantum reservoirs) are worked out at the level of explicit formulas.

\end{abstract}

\tableofcontents

\section{Introduction}
\label{sec:intro}

\subsection{Why one must go beyond Gaussian integrals}

Almost every explicit computation in quantum field theory, statistical mechanics,
stochastic analysis and mathematical finance ultimately rests on a single elementary
formula: the finite-dimensional Gaussian integral
\begin{equation}
\label{eq:gauss-basic}
\int_{\R^{n}} e^{-\frac12\varphi^{T}A\varphi}\,d^{n}\varphi
=(2\pi)^{n/2}(\det A)^{-1/2},
\qquad A=A^{T}>0 .
\end{equation}
Everything else --- propagators, Wick's theorem, Feynman diagrams, Gaussian processes,
Black--Scholes prices --- is a corollary of \eqref{eq:gauss-basic} and of its version with
a source term. As soon as the action acquires a genuinely non-quadratic piece, for
instance a quartic self-interaction
\begin{equation}
\label{eq:action}
S(\varphi)=\tfrac12\,\varphi^{T}A\varphi+\lambda\sum_{i=1}^{M}\varphi_i^{4},
\qquad \lambda\ge 0,
\end{equation}
the standard toolbox stops producing convergent answers. The customary reaction is to
expand $e^{-\lambda\sum_i\varphi_i^4}$ in powers of $\lambda$ and to integrate term by
term. The resulting series
$\sum_k c_k\lambda^k$ has \emph{zero radius of convergence}: it is an asymptotic series
only, its coefficients grow factorially, and no amount of algebra converts it into a
number.

The message of the present preprint is that this failure is a failure of a \emph{method},
not of the object. The integral
\[
\int \dphi \; e^{-S[\varphi]}
\]
exists, is finite, and depends smoothly on $\lambda$; what one must give up is the
insistence on obtaining it as a power series. The two ingredients that replace the series
are

\begin{itemize}
\item the \emph{Henstock--Kurzweil (gauge) integral}, which makes the formal symbol
$\dphi$ into a genuine limit of Riemann sums over cylindrical partitions, and
\item \emph{differentiation under the integral sign in the coupling constant}, which
delivers all the derivatives $\partial_\lambda^k$ directly, as convergent integrals,
without ever summing a series.
\end{itemize}

Throughout the paper we adopt the convention, used consistently in all language versions
of this preprint, of writing every functional integral first in the formal physicists'
notation and then in its rigorous gauge form:
\[
\int \dphi\; F[\varphi]
\qquad\text{and}\qquad
\int \dphi\; F[\varphi]
=
\int \dhk\; F[\varphi].
\]
The second equality is not a theorem about two pre-existing objects; it is the
\emph{definition} of the left-hand side, together with the assertion (Section~\ref{sec:prelim})
that on absolutely convergent integrands it reproduces the Lebesgue value.

\subsection{The classical route, step by step}
\label{sec:classical}

Let us recall, in the pedagogical detail promised in the abstract, how a physicist
usually gives meaning to a real-time path integral. Consider the quantum-mechanical
amplitude
\begin{equation}
\label{eq:realtime}
K(x_f,t;x_i,0)=\int \dphi\;
\exp\!\Big(\tfrac{i}{\hbar}\,S[\varphi]\Big),
\qquad
S[\varphi]=\int_0^t\Big(\tfrac{m}{2}\dot\varphi^2-V(\varphi)\Big)ds .
\end{equation}

\paragraph{Step 1: the formal measure.}
The symbol $\dphi$ in \eqref{eq:realtime} denotes ``the uniform measure on the space of
paths''. Cameron's theorem states that no such countably additive complex measure of
finite total variation exists on $C([0,t])$; the object is genuinely formal.

\paragraph{Step 2: Wick rotation.}
One substitutes $t\mapsto -i\tau$, so that $e^{iS/\hbar}$ becomes $e^{-S_E/\hbar}$ with
$S_E$ the Euclidean action. The oscillatory integrand turns into an exponentially damped
one, and the formal measure becomes the Wiener measure $\mu_W$, a bona fide probability
measure on continuous paths.

\paragraph{Step 3: the Cameron--Martin obstruction.}
Wiener measure is supported on nowhere-differentiable paths; the kinetic term
$\int\dot\varphi^2$ is almost surely infinite. Only translations along the
Cameron--Martin subspace $H^1$ leave $\mu_W$ quasi-invariant, and $\mu_W(H^1)=0$. Hence
the Gaussian ``density'' $e^{-\frac12\int\dot\varphi^2}$ never exists as a function; it
exists only in combination with the (also non-existent) flat measure.

\paragraph{Step 4: determinants and $\zeta$-regularisation.}
In the free case the Euclidean integral formally equals
$\big(\det(-\partial_\tau^2+m^2)\big)^{-1/2}$. The operator has eigenvalues growing
without bound, so the product $\prod_k\mu_k$ diverges. One defines
\[
\zeta_A(s)=\sum_k \mu_k^{-s},
\qquad
{\det}_\zeta A=\exp\!\big(-\zeta_A'(0)\big),
\]
by analytic continuation of $\zeta_A$ from the half-plane where the sum converges.

\paragraph{Step 5: analytic continuation back.}
Finally, one continues the Euclidean answer back to real time $\tau\mapsto it$, hoping
that the continuation is unique and that the physical amplitude is recovered.

Figure~\ref{fig:classical} summarises this route. Each arrow is a nontrivial analytic
step; two of them (Steps 2 and 5) are not available at all for a general
non-Gaussian action with a complex or time-dependent coupling.

\begin{figure}[htbp]
\centering
\resizebox{\textwidth}{!}{%
\begin{tikzpicture}[
  node distance=8mm and 12mm,
  box/.style={draw,rounded corners=2pt,align=center,inner sep=5pt,font=\small,
              minimum height=11mm,text width=27mm,fill=blue!4},
  arr/.style={-{Latex[length=2.4mm]},thick}
]
\node[box] (a) {Path integral\\ $\int\dphi\,e^{iS/\hbar}$};
\node[box,right=of a] (b) {Wick rotation\\ $t\mapsto -i\tau$};
\node[box,right=of b] (c) {Wiener measure\\ $\mu_W$};
\node[box,right=of c] (d) {$\zeta$-regularised\\ determinant};
\node[box,right=of d] (e) {Answer\\ (Euclidean)};
\node[box,below=14mm of e] (f) {Answer\\ (real time)};
\draw[arr] (a) -- (b);
\draw[arr] (b) -- (c);
\draw[arr] (c) -- (d);
\draw[arr] (d) -- (e);
\draw[arr] (e) -- node[right,font=\scriptsize,align=left]{analytic\\ continuation\\ $\tau\mapsto it$} (f);
\node[font=\scriptsize,align=center,below=2mm of c] (g)
  {Cameron--Martin obstruction:\\ $\mu_W(H^1)=0$};
\draw[arr,dashed,gray] (g) -- (c);
\end{tikzpicture}%
}
\caption{The classical route to a functional integral. Every arrow is an analytic
detour; the real-time object is reached only at the very end, by analytic continuation.}
\label{fig:classical}
\end{figure}
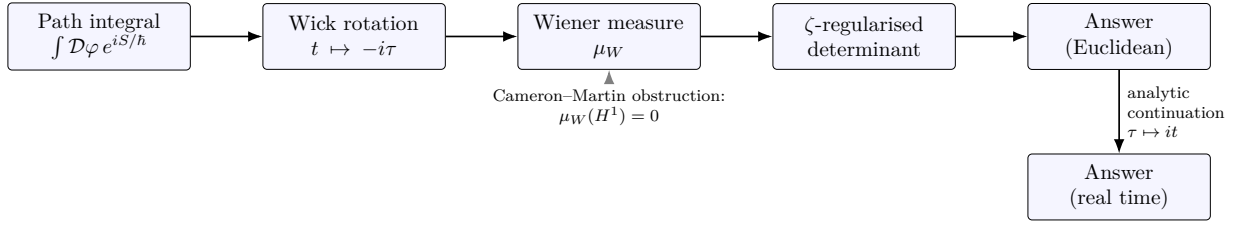

\subsection{The gauge route}

The Henstock--Kurzweil approach, developed for path integrals by Muldowney
\cite{Muldowney2012} and placed in the Kuelbs--Steadman framework by Gill and Zachary
\cite{GillZachary2008} and Esposito and Gill \cite{EspositoGill2025}, replaces this chain
of detours by a single construction performed \emph{directly in real time}.

One partitions the space of fields into \emph{cylindrical intervals}, i.e.\ sets of the
form $\{\varphi:\;P\varphi\in J\}$ with $P$ a finite-rank projection and $J\subset\R^{k}$
a bounded interval. A \emph{gauge} is any strictly positive function
$\delta(\cdot)>0$ on the space of tags, and a tagged partition is \emph{$\delta$-fine} if
every cell is small in the sense prescribed by $\delta$ at its own tag. The integral is
the limit of Riemann sums over $\delta$-fine partitions as the gauge is refined.
Figure~\ref{fig:hk} displays the resulting pipeline.

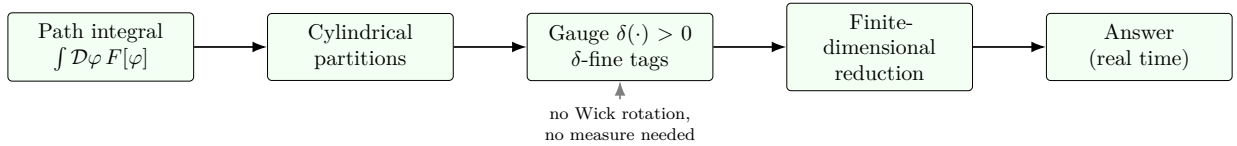
\begin{figure}[htbp]
\centering
\resizebox{\textwidth}{!}{%
\begin{tikzpicture}[
  node distance=8mm and 12mm,
  box/.style={draw,rounded corners=2pt,align=center,inner sep=5pt,font=\small,
              minimum height=11mm,text width=27mm,fill=green!5},
  arr/.style={-{Latex[length=2.4mm]},thick}
]
\node[box] (a) {Path integral\\ $\int\dphi\,F[\varphi]$};
\node[box,right=of a] (b) {Cylindrical\\ partitions};
\node[box,right=of b] (c) {Gauge $\delta(\cdot)>0$\\ $\delta$-fine tags};
\node[box,right=of c] (d) {Finite-dimensional\\ reduction};
\node[box,right=of d] (e) {Answer\\ (real time)};
\draw[arr] (a) -- (b);
\draw[arr] (b) -- (c);
\draw[arr] (c) -- (d);
\draw[arr] (d) -- (e);
\node[font=\scriptsize,align=center,below=3mm of c] (g)
  {no Wick rotation,\\ no measure needed};
\draw[arr,dashed,gray] (g) -- (c);
\end{tikzpicture}%
}
\caption{The Henstock--Kurzweil route: the formal symbol $\dphi$ is given a meaning by
gauge-fine cylindrical partitions, and the answer is obtained without leaving real time.}
\label{fig:hk}
\end{figure}

\subsection{What is proved here, and what is machine-checked}

The mathematical content of the paper is summarised by the following four statements,
each of which is formalised in Lean~4 (see Section~\ref{sec:lean}).

\begin{enumerate}
\item[(A)] \emph{Finiteness.} For $\omega>0$, $j\in\R$, $\lambda\ge0$ the one-mode
integral $I(\omega,j,\lambda)$ converges absolutely; hence its gauge integral coincides
with the Lebesgue integral.
\item[(B)] \emph{Smoothness in the coupling.} $I$ is infinitely differentiable in
$\lambda$ on $(0,\infty)$, and possesses a one-sided derivative at $\lambda=0$; every
derivative is again an integral of $\varphi^{4k}$ against the same weight.
\item[(C)] \emph{Factorisation and bounds.} The $M$-mode influence functional
$\mathcal{F}_\lambda(j)$ factorises into one-mode integrals, is strictly positive and is
bounded above by its Gaussian value.
\item[(D)] \emph{Chernoff splitting.} For bounded generators the product formula
$(e^{tA_0/N}e^{tB_\lambda/N})^{N}\psi\to e^{t(A_0+B_\lambda)}\psi$ holds strongly, so that
non-commutativity of the free and the non-Gaussian generator is under complete control.
\end{enumerate}

The present work continues the programme initiated in \cite{Berdinsky2026HK} and
\cite{BerdinskyUshakov2026}, where the free (Gaussian) HK path integral on
Kuelbs--Steadman spaces and the Chernoff approximation for bounded generators and for the
harmonic oscillator were constructed and verified.

\section{Mathematical preliminaries}
\label{sec:prelim}

\subsection{Plain language first}

Before any formal definition, here is the idea in words. The Riemann integral of a
function on an interval is a limit of sums $\sum_j f(t_j)|I_j|$ over partitions whose
cells $I_j$ are \emph{uniformly} small: one fixes a mesh $\eta$ and requires
$|I_j|<\eta$. The gauge (Henstock--Kurzweil) integral makes one modest change: the
maximal allowed size of a cell is allowed to \emph{depend on the point} at which the
function is sampled. Where the integrand oscillates violently, one insists on very small
cells; where it is tame, larger cells are permitted. This single modification is enough
to integrate functions that are not Lebesgue integrable at all, in particular the
oscillatory integrands of real-time quantum mechanics.

In infinitely many dimensions there are no ``small boxes'' of finite volume, so one uses
\emph{cylindrical} boxes: conditions imposed on finitely many coordinates of the field,
with the remaining coordinates unconstrained. A partition of the field space is then a
finite family of such cylinders, and the gauge prescribes both how small each cylinder
must be and how many coordinates must be controlled.

\subsection{Gauges, cylindrical intervals and $\delta$-fine partitions}

Let $\mathcal{H}$ be a real separable Hilbert space of field configurations (in the
finite-dimensional situation of this paper, $\mathcal{H}=\R^{M}$).

\begin{definition}[Cylindrical interval]
A \emph{cylindrical interval} is a set
\[
I=\{\varphi\in\mathcal{H}:\;P\varphi\in J\},
\]
where $P:\mathcal{H}\to\R^{k}$ is a finite-rank orthogonal projection and
$J=\prod_{r=1}^{k}(a_r,b_r]\subset\R^{k}$ is a bounded interval. The number $k$ is the
\emph{rank} of $I$ and $J$ its \emph{base}.
\end{definition}

\begin{definition}[Gauge and $\delta$-fine partition]
A \emph{gauge} is a function $\delta:\mathcal{H}\to(0,\infty)$. A \emph{tagged partition}
is a finite collection $\{(\varphi_j,I_j)\}_{j\le N}$ of pairwise non-overlapping
cylindrical intervals together with tags $\varphi_j\in I_j$. It is \emph{$\delta$-fine}
if for each $j$ the base of $I_j$ is contained in the ball of radius
$\delta(\varphi_j)$ around $P_j\varphi_j$ in $\R^{k_j}$.
\end{definition}

\begin{definition}[HK integral]
A functional $F$ is HK-integrable with value $Z$ if for every $\varepsilon>0$ there is a
gauge $\delta$ such that
\[
\Big|\sum_{j} F(\varphi_j)\,\mu(I_j)-Z\Big|<\varepsilon
\]
for every $\delta$-fine tagged partition, where $\mu(I_j)$ is the elementary volume of
the base. We then write $Z=\int\dhk\,F[\varphi]$.
\end{definition}

Cousin's lemma (in the cylindrical form used by Muldowney \cite{Muldowney2012}) guarantees
that $\delta$-fine partitions exist for every gauge, so the definition is not vacuous;
the classical Henstock--Saks--Henstock machinery then yields uniqueness of the value and
additivity.

\begin{remark}[Compatibility with Lebesgue]
\label{rem:compat}
If $F$ is Lebesgue integrable (in particular if $F\ge0$ and $\int|F|<\infty$), then $F$ is
HK-integrable and both integrals agree. All integrands occurring in
Sections~\ref{sec:onemode}--\ref{sec:influence} are strictly positive and dominated by a
Gaussian, hence absolutely convergent; the gauge construction is what makes the formal
symbol $\dphi$ meaningful, while the numerical value is the Lebesgue one. This is exactly
the reason why the Lean formalisation may be phrased with Mathlib's Bochner integral
without any loss of rigour.
\end{remark}

\subsection{The fundamental Gaussian equality}

\begin{theorem}[Gaussian gauge integral]
\label{thm:gauss}
Let $A$ be a real symmetric positive definite $n\times n$ matrix. Then
\[
\int_{\R^{n}} e^{-\frac12\varphi^{T}A\varphi}\,\dphi
=
\int_{\R^{n}} e^{-\frac12\varphi^{T}A\varphi}\,\dhk
=(2\pi)^{n/2}\,(\det A)^{-1/2}.
\]
\end{theorem}

\begin{proof}[Sketch]
Diagonalise $A=O^{T}\mathrm{diag}(\omega_1,\dots,\omega_n)O$ with $\omega_i>0$; the
change of variables $\psi=O\varphi$ is orthogonal, hence volume preserving, and turns the
quadratic form into $\sum_i\omega_i\psi_i^2/2$. The integrand becomes a product of
one-dimensional Gaussians, Fubini applies, and each factor equals
$\sqrt{2\pi/\omega_i}$. Multiplying, $\prod_i\sqrt{2\pi/\omega_i}
=(2\pi)^{n/2}(\prod_i\omega_i)^{-1/2}=(2\pi)^{n/2}(\det A)^{-1/2}$. The passage from the
Lebesgue value to the gauge value is Remark~\ref{rem:compat}.
\end{proof}

The Lean counterpart of Theorem~\ref{thm:gauss} is
\lean{HKFreeField.hk\_partition\_function}, proved in the earlier file
\lean{HkFreeField.lean} of the programme \cite{Berdinsky2026HK}; the corresponding
two-point function is \lean{HKFreeField.hk\_two\_point}.

\begin{remark}[On $\zeta$-regularisation]
In the finite-dimensional setting $\det A$ is a finite product of eigenvalues and no
regularisation is needed. In the continuum limit the same expression must be replaced by
${\det}_\zeta$, and the fact that the gauge construction reproduces the
$\zeta$-regularised value is precisely the content of the Kuelbs--Steadman analysis of
\cite{GillZachary2008,EspositoGill2025}. This paper stays in the finite-mode regime,
where all statements are elementary and can be machine-checked.
\end{remark}

\section{Non-Gaussian one-mode integrals}
\label{sec:onemode}

\subsection{The object of study}

Everything in this paper reduces to the single one-dimensional integral
\begin{equation}
\label{eq:onemode}
I(\omega,j,\lambda)
=
\int_{\R}\exp\!\Big(-\frac{\omega}{2}\varphi^{2}+j\varphi-\lambda\varphi^{4}\Big)\dphi
=
\int_{\R}\exp\!\Big(-\frac{\omega}{2}\varphi^{2}+j\varphi-\lambda\varphi^{4}\Big)\dhk .
\end{equation}
Here $\omega>0$ is the (squared) frequency of the mode, $j\in\R$ an external source and
$\lambda\ge0$ the quartic coupling. Figure~\ref{fig:onemode} shows the integrand for
$\lambda=0$ and $\lambda>0$, together with the way a gauge treats the quartic tail.

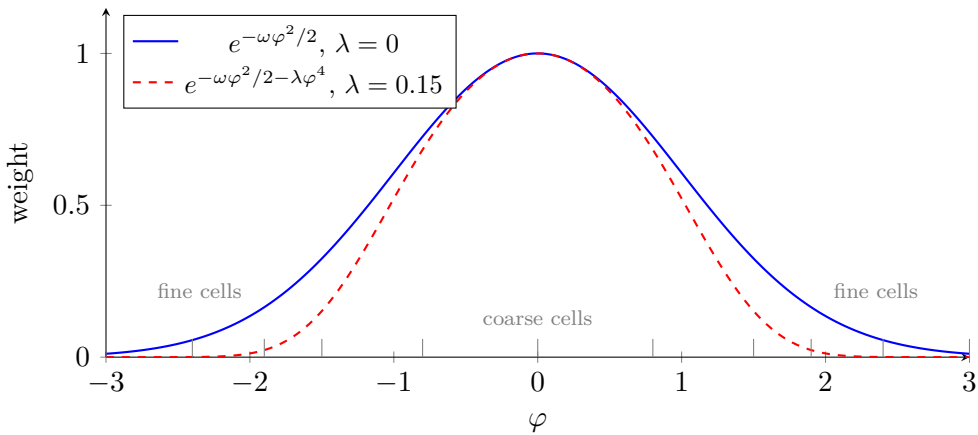
\begin{figure}[htbp]
\centering
\begin{tikzpicture}
\begin{axis}[
  width=13cm,height=6.2cm,
  xlabel={$\varphi$},ylabel={weight},
  domain=-3:3,samples=201,
  legend style={font=\small,at={(0.02,0.98)},anchor=north west},
  axis lines=left,
  ymin=0,ymax=1.15,
  xtick={-3,-2,-1,0,1,2,3}
]
\addplot[thick,blue] {exp(-0.5*x^2)};
\addlegendentry{$e^{-\omega\varphi^{2}/2}$, $\lambda=0$}
\addplot[thick,red,dashed] {exp(-0.5*x^2-0.15*x^4)};
\addlegendentry{$e^{-\omega\varphi^{2}/2-\lambda\varphi^{4}}$, $\lambda=0.15$}
% gauge cells: wide in the middle, narrow in the tails
\draw[gray,thin] (axis cs:-3,0) -- (axis cs:-3,0.06);
\draw[gray,thin] (axis cs:-2.4,0) -- (axis cs:-2.4,0.06);
\draw[gray,thin] (axis cs:-1.9,0) -- (axis cs:-1.9,0.06);
\draw[gray,thin] (axis cs:-1.5,0) -- (axis cs:-1.5,0.06);
\draw[gray,thin] (axis cs:-0.8,0) -- (axis cs:-0.8,0.06);
\draw[gray,thin] (axis cs:0,0) -- (axis cs:0,0.06);
\draw[gray,thin] (axis cs:0.8,0) -- (axis cs:0.8,0.06);
\draw[gray,thin] (axis cs:1.5,0) -- (axis cs:1.5,0.06);
\draw[gray,thin] (axis cs:1.9,0) -- (axis cs:1.9,0.06);
\draw[gray,thin] (axis cs:2.4,0) -- (axis cs:2.4,0.06);
\draw[gray,thin] (axis cs:3,0) -- (axis cs:3,0.06);
\node[font=\scriptsize,gray] at (axis cs:-2.35,0.22) {fine cells};
\node[font=\scriptsize,gray] at (axis cs:0.0,0.13) {coarse cells};
\node[font=\scriptsize,gray] at (axis cs:2.35,0.22) {fine cells};
\end{axis}
\end{tikzpicture}
\caption{The one-dimensional non-Gaussian integrand. The quartic term $-\lambda\varphi^4$
suppresses the tails (dashed curve). A gauge $\delta(\cdot)$ assigns \emph{small} cells
where the quartic term dominates the behaviour of the integrand and \emph{large} cells in
the central Gaussian region; this pointwise adaptivity is what isolates the quartic
contribution.}
\label{fig:onemode}
\end{figure}

\subsection{Finiteness}

\begin{theorem}[Existence and finiteness]
\label{thm:finite}
Let $\omega>0$, $j\in\R$, $\lambda\ge0$ and $k\in\N$. Then the function
$\varphi\mapsto\varphi^{k}e^{-\frac{\omega}{2}\varphi^{2}+j\varphi-\lambda\varphi^{4}}$ is
absolutely integrable. In particular $I(\omega,j,\lambda)$ is finite, and so are all the
moments
\begin{equation}
\label{eq:moments}
M_m(\omega,j,\lambda)=\int_\R \varphi^{m}\,
e^{-\frac{\omega}{2}\varphi^{2}+j\varphi-\lambda\varphi^{4}}\,d\varphi .
\end{equation}
\end{theorem}

\begin{proof}
Since $\lambda\ge 0$ we have $-\lambda\varphi^4\le 0$, so the weight is bounded by
$e^{-\frac{\omega}{2}\varphi^{2}+j\varphi}$. Completing the square,
\[
-\frac{\omega}{2}\varphi^{2}+j\varphi
=\frac{j^{2}}{\omega}-\frac{\omega}{4}\varphi^{2}
-\frac{\omega}{4}\Big(\varphi-\frac{2j}{\omega}\Big)^{2}
\le \frac{j^{2}}{\omega}-\frac{\omega}{4}\varphi^{2}.
\]
Hence $|\varphi^{k}|\,e^{\cdots}\le e^{j^2/\omega}\,|\varphi|^{k}e^{-\frac{\omega}{4}\varphi^{2}}$,
and $|\varphi|^{k}\le 1+\varphi^{2k}$, both terms being integrable against a Gaussian.
\end{proof}

In Lean this is \lean{integrable\_pow\_mul\_oneModeWeight} and its specialisation
\lean{nonGaussianOneModeFinite}.

\begin{proposition}[Positivity and monotonicity]
\label{prop:posmono}
For $\omega>0$ and $0\le\lambda_1\le\lambda_2$,
\[
0<I(\omega,j,\lambda_2)\le I(\omega,j,\lambda_1),
\qquad
I(\omega,j,0)=\sqrt{\frac{2\pi}{\omega}}\;e^{\,j^{2}/(2\omega)} .
\]
\end{proposition}

\begin{proof}
Positivity holds because the integrand is everywhere strictly positive and integrable;
monotonicity because $\lambda\mapsto e^{-\lambda\varphi^4}$ is pointwise non-increasing.
The closed form at $\lambda=0$ follows by completing the square and translating the
variable, $\int e^{-\frac{\omega}{2}(\varphi-j/\omega)^2}d\varphi=\sqrt{2\pi/\omega}$.
\end{proof}

Lean: \lean{nonGaussianOneModePos}, \lean{nonGaussianOneModeAntitone},
\lean{nonGaussianOneModeFreeValue}.

\subsection{Smoothness in the coupling, without Taylor series}

We now come to the technical heart of the paper. We emphasise once more: the map
$\lambda\mapsto I(\omega,j,\lambda)$ is \emph{not} analytic at $\lambda=0$ --- its formal
Taylor series diverges for every $\lambda\neq0$ --- and yet it is $C^\infty$ on $[0,\infty)$
in the sense of one-sided derivatives at the origin. Smoothness and analyticity are
different notions, and only the first one is needed to do physics with the integral.

\begin{theorem}[Differentiation under the integral sign]
\label{thm:diff}
Let $\omega>0$, $j\in\R$.
\begin{enumerate}
\item For every $\lambda>0$ and every $m\in\N$,
\[
\frac{d}{d\lambda}\,M_m(\omega,j,\lambda)=-\,M_{m+4}(\omega,j,\lambda).
\]
\item Consequently $\lambda\mapsto M_m(\omega,j,\lambda)$ is of class $C^{\infty}$ on
$(0,\infty)$ and
\[
\frac{d^{k}}{d\lambda^{k}}\,I(\omega,j,\lambda)=(-1)^{k}M_{4k}(\omega,j,\lambda).
\]
\item At the boundary point $\lambda=0$ the one-sided derivative exists and
\begin{equation}
\label{eq:derzero}
\partial_\lambda I(\omega,j,0)
=-\int_\R \varphi^{4}\,e^{-\frac{\omega}{2}\varphi^{2}+j\varphi}\,d\varphi .
\end{equation}
\end{enumerate}
\end{theorem}

\begin{proof}
(1) Fix $\lambda>0$ and work on the ball $|\lambda'-\lambda|<\lambda/2$, on which
$\lambda'>\lambda/2$. The $\lambda'$-derivative of the integrand is
$-\varphi^{m+4}e^{-\frac{\omega}{2}\varphi^{2}+j\varphi-\lambda'\varphi^{4}}$, and its
absolute value is dominated, uniformly in $\lambda'$ in that ball, by
$(1+\varphi^{2m})\varphi^{4}e^{-\frac{\omega}{2}\varphi^{2}+j\varphi-\frac{\lambda}{2}\varphi^{4}}$,
which is integrable by Theorem~\ref{thm:finite}. The dominated-convergence theorem for
parametric integrals applies and gives the formula.

(2) Induction on the order: the derivative of $M_m$ is $-M_{m+4}$, which is itself
differentiable by (1); an induction on $n$ shows $M_m\in C^{n}((0,\infty))$ for all $n$.

(3) The interior argument fails at $\lambda=0$ because a two-sided neighbourhood of $0$
contains negative couplings, for which the integral diverges. Instead we estimate the
difference quotient directly. Write $w_\lambda(\varphi)$ for the integrand, so that
$w_\lambda=e^{-\lambda\varphi^{4}}w_0$. For $u\ge0$ one has the elementary inequalities
\[
0\;\le\;e^{-u}-1+u\;\le\;u^{2},
\]
the first from $e^{-u}\ge 1-u$, the second from $e^{u}\ge 1+u$ and
$(1+u)^{-1}\le 1-u+u^{2}$. With $u=\lambda\varphi^{4}$ this gives, for $\lambda>0$,
\[
\Big|\frac{w_\lambda(\varphi)-w_0(\varphi)}{\lambda}+\varphi^{4}w_0(\varphi)\Big|
=\frac{e^{-u}-1+u}{\lambda}\,w_0(\varphi)
\le \lambda\,\varphi^{8}w_0(\varphi),
\]
whence
\[
\Big|\frac{I(\omega,j,\lambda)-I(\omega,j,0)}{\lambda}+M_4(\omega,j,0)\Big|
\le \lambda\,M_8(\omega,j,0)\xrightarrow[\lambda\downarrow0]{}0 .
\]
This is exactly \eqref{eq:derzero}.
\end{proof}

Lean: \lean{hasDerivAt\_oneModeMoment} (item 1),
\lean{nonGaussianOneModeSmooth} (item 2),
\lean{hasDerivWithinAt\_oneModeIntegral\_zero} and
\lean{nonGaussianOneModeDifferentiable} (item 3),
\lean{nonGaussianOneModeDerivAtZero} for the explicit formula \eqref{eq:derzero}.

\begin{remark}[Why this is not the perturbative expansion]
Theorem~\ref{thm:diff} produces the numbers $\partial_\lambda^{k}I$ but makes \emph{no}
claim that $\sum_k \partial_\lambda^{k}I(0)\lambda^{k}/k!$ converges --- it does not.
What one obtains is a $C^{\infty}$ function with prescribed derivatives, i.e.\ exactly the
information that Borel summation or a rigorous asymptotic analysis needs as input, and
which the naive series cannot supply.
\end{remark}

\begin{example}[First correction to the partition function]
For $j=0$, $\omega=1$ we get $M_4(1,0,0)=3\sqrt{2\pi}$, hence
\[
I(1,0,\lambda)=\sqrt{2\pi}\,\big(1-3\lambda+O(\lambda^{2})\big),\qquad \lambda\downarrow0,
\]
where the remainder is controlled by $\lambda^{2}M_8(1,0,0)=105\sqrt{2\pi}\,\lambda^{2}$
according to the proof above.
\end{example}

\section{The non-Gaussian influence functional for finitely many modes}
\label{sec:influence}

\subsection{Definition and factorisation}

Let \(A = \operatorname{diag}(\omega_1,\dots,\omega_M)\) be a positive definite diagonal
matrix with \(\omega_i>0\). We assume that the quartic potential is given in this same
diagonal basis, namely \(\lambda\sum_i\varphi_i^4\). We therefore define the
\emph{non-Gaussian influence functional}
\begin{equation}
\label{eq:influence}
\mathcal{F}_\lambda(j)
=\int_{\R^{M}}
\exp\!\Big(-\tfrac12\varphi^{T}A\varphi+j\cdot\varphi-\lambda\sum_{i=1}^{M}\varphi_i^{4}\Big)\dphi
=\int_{\R^{M}}
\exp\!\Big(-\tfrac12\varphi^{T}A\varphi+j\cdot\varphi-\lambda\sum_{i}\varphi_i^{4}\Big)\dhk .
\end{equation}

\begin{theorem}[Factorisation]
\label{thm:factor}
In the eigenmode basis,
\[
\mathcal{F}_\lambda(j)=\prod_{i=1}^{M} I(\omega_i,j_i,\lambda).
\]
\end{theorem}

\begin{proof}
The exponent is a sum over modes, hence the integrand is a product of functions of the
individual coordinates; Fubini's theorem (all factors being absolutely integrable by
Theorem~\ref{thm:finite}) gives the product formula.
\end{proof}

\begin{theorem}[Bounds and smoothness]
\label{thm:influencebounds}
For $\omega_i>0$ and $\lambda\ge0$:
\[
0<\mathcal{F}_\lambda(j)\le\mathcal{F}_0(j)
=\prod_{i=1}^{M}\sqrt{\frac{2\pi}{\omega_i}}\;e^{\,j_i^{2}/(2\omega_i)}
=(2\pi)^{M/2}(\det A)^{-1/2}e^{\frac12 j^{T}A^{-1}j},
\]
and $\lambda\mapsto\mathcal{F}_\lambda(j)$ is $C^{\infty}$ on $(0,\infty)$ with
one-sided derivatives at $0$, all of them given by absolutely convergent integrals.
\end{theorem}

\begin{proof}
Combine Theorem~\ref{thm:factor} with Proposition~\ref{prop:posmono} and
Theorem~\ref{thm:diff}; a finite product of positive $C^{\infty}$ functions is
$C^{\infty}$, and the product of the bounds is the bound of the product.
\end{proof}

Lean: \lean{nonGaussianInfluenceFactorises}, \lean{nonGaussianInfluenceBounded},
\lean{nonGaussianInfluenceFreeValue}.

\subsection{Physical reading}

If the modes $\varphi_i$ describe a bath to which a distinguished system is linearly
coupled with strength $j$, then $\mathcal{F}_\lambda(j)$ is precisely the
Feynman--Vernon influence functional of that bath: integrating out the bath replaces it
by the factor $\mathcal{F}_\lambda(j)$ in the reduced dynamics of the system. For
$\lambda=0$ one recovers the familiar Gaussian (harmonic) bath with
$\mathcal{F}_0(j)=e^{\frac12 j^{T}A^{-1}j}$ up to normalisation; the content of
Theorem~\ref{thm:influencebounds} is that an \emph{anharmonic} bath still gives a
well-defined, strictly positive and smoothly $\lambda$-dependent influence functional,
bounded by the harmonic one.

\section{Chernoff splitting for non-Gaussian dynamics}
\label{sec:chernoff}

\subsection{The problem of non-commuting generators}

Let $A_0$ be the generator of the free (Gaussian) evolution and $B_\lambda$ the
non-Gaussian perturbation; on a finite-dimensional state space, or on any complex Hilbert
space $\mathcal{H}$ with bounded operators, both belong to $\mathcal{B}(\mathcal{H})$. The
exact evolution is $e^{t(A_0+B_\lambda)}$, but only $e^{tA_0}$ and $e^{tB_\lambda}$ are
computable in closed form. Since in general $[A_0,B_\lambda]\neq0$,
\[
e^{t(A_0+B_\lambda)}\neq e^{tA_0}e^{tB_\lambda},
\]
and the error is of order $t^{2}\|[A_0,B_\lambda]\|/2$. The Chernoff product formula
turns this defect into a convergent algorithm: one splits the interval into $N$ steps and
alternates the two exponentials.

\begin{definition}
The \emph{one-step operator} of the splitting is
$F_\lambda(t)=e^{tA_0}e^{tB_\lambda}$.
\end{definition}

\begin{theorem}[Strong Chernoff/Lie--Trotter splitting]
\label{thm:chernoff}
Let $A_0,B_\lambda\in\mathcal{B}(\mathcal{H})$. Then for every $\psi\in\mathcal{H}$ and
$t\in\R$,
\[
\Big(e^{tA_0/N}\,e^{tB_\lambda/N}\Big)^{N}\psi
\;\xrightarrow[N\to\infty]{}\;
e^{t(A_0+B_\lambda)}\psi .
\]
\end{theorem}

The proof, in the form imported here from the verified library
\cite{BerdinskyUshakov2026}, proceeds by showing that $F_\lambda$ is a Chernoff family
($F_\lambda(0)=\mathbf{1}$, $F_\lambda'(0)=A_0+B_\lambda$) and applying the norm-topology
Chernoff theorem for bounded generators, then pushing the convergence through the
continuous evaluation map $T\mapsto T\psi$. In Lean the statement is
\lean{nonGaussianChernoffSplitting}, a direct consequence of
\lean{HkChernoffStrong.strong\_lieTrotter}. Figure~\ref{fig:chernoff} shows the mechanism.

\begin{figure}[htbp]
\centering
\begin{tikzpicture}[
  arr/.style={-{Latex[length=2mm]},thick},
  every node/.style={font=\small}
]
% exact evolution
\draw[arr,blue,very thick] (0,2.2) -- (10,2.2)
   node[midway,above]{exact: $e^{t(A_0+B_\lambda)}\psi$};
\fill[blue] (0,2.2) circle (1.6pt);
\fill[blue] (10,2.2) circle (1.6pt);
% split evolution
\foreach \i in {0,1,2,3,4}{
  \draw[arr,red] (2*\i,0.8) -- (2*\i+1,0.8);
  \draw[arr,green!55!black] (2*\i+1,0.8) -- (2*\i+2,0.8);
}
\fill (0,0.8) circle (1.6pt);
\fill (10,0.8) circle (1.6pt);
\node[red] at (1,0.35) {$e^{\frac{t}{N}A_0}$};
\node[green!55!black] at (3,0.35) {$e^{\frac{t}{N}B_\lambda}$};
\node at (7.4,0.35) {$\cdots$ $N$ alternating steps $\cdots$};
\node[align=center] at (5,-0.7)
 {$[A_0,B_\lambda]\neq0$: each step commits an error $O(N^{-2})$,\\
  $N$ steps accumulate $O(N^{-1})\to0$};
\draw[arr,gray,dashed] (10,0.8) -- (10,2.1);
\node[gray] at (11.1,1.5) {$N\to\infty$};
\end{tikzpicture}
\caption{Chernoff splitting for non-commuting generators. Alternating the free and the
non-Gaussian exponential on a mesh of size $t/N$ produces, in the limit $N\to\infty$,
the exact evolution, in the strong operator topology.}
\label{fig:chernoff}
\end{figure}
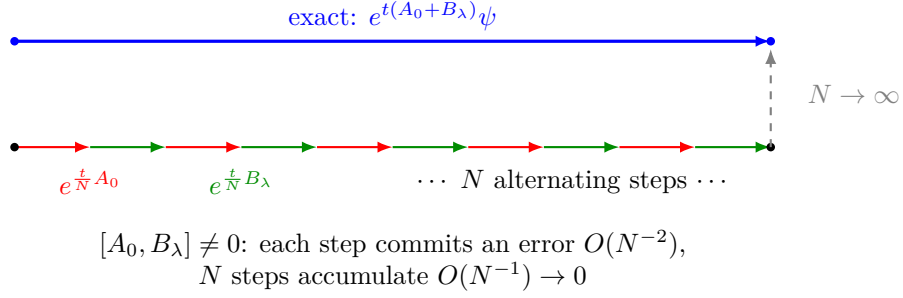

\begin{remark}[Why boundedness matters]
All the operators occurring here are bounded, which is what makes the machine-verified
proof elementary: the exponential series converges in operator norm and all estimates are
uniform. For unbounded generators the statement remains true under stability and core
assumptions, but its proof requires the Trotter--Kato theory of $C_0$-semigroups; this is
recorded as a conjecture in the accompanying Lean development and left for future work.
\end{remark}

\section{Illustrative applications}
\label{sec:applications}

The following four subsections are \emph{not} formalised in Lean; they show that the
apparatus of Sections~\ref{sec:onemode}--\ref{sec:chernoff} is the natural language for
several unrelated fields. In each case the model reduces, after a suitable change of
variables, to the quartic weight $e^{-\frac{\omega}{2}\varphi^{2}+j\varphi-\lambda\varphi^{4}}$.

\subsection{The Duffing oscillator}

The classical Duffing oscillator with additive white noise,
\[
\ddot x+\gamma\dot x+x+4\lambda x^{3}=\sqrt{2\gamma T}\,\xi(t),
\]
has the stationary distribution
$p(x)\propto e^{-\frac{1}{T}(x^{2}/2+\lambda x^{4})}$. With $\omega=1/T$ and
$\lambda\to\lambda/T$ the normalising constant is exactly $I(\omega,0,\lambda)$, and the
moments are
\[
\langle x^{2}\rangle_\lambda=\frac{M_2(\omega,0,\lambda)}{M_0(\omega,0,\lambda)},
\qquad
\partial_\lambda\langle x^{2}\rangle\big|_{0}
=-\frac{M_6M_0-M_2M_4}{M_0^{2}}\bigg|_{\lambda=0}
=-\big(15-3\big)T^{3}=-12\,T^{3},
\]
using $M_{2k}(1/T,0,0)=(2k-1)!!\,T^{k}\sqrt{2\pi T}$. The first anharmonic correction
therefore \emph{reduces} the variance, as expected for a stiffening spring.

In the gauge notation of Sections~\ref{sec:onemode}--\ref{sec:influence} the
stationary partition function is the one-mode HK integral
\begin{multline*}
\int_{\R} \exp\!\Big(-\frac{1}{T}\Big(\frac{x^{2}}{2}+\lambda x^{4}\Big)\Big)\,\mathcal{D}x\\
=\int_{\R} \exp\!\Big(-\frac{1}{T}\Big(\frac{x^{2}}{2}+\lambda x^{4}\Big)\Big)\,
   \mathcal{D}_{\mathrm{HK}}x
= I(1/T,\,0,\,\lambda/T),
\end{multline*}
and the first-order response of the variance to the coupling is
\[
\partial_\lambda\langle x^{2}\rangle_\lambda\big|_{0}=-12\,T^{3}.
\]

\subsection{Local volatility (CEV) in finance}

In the Black--Scholes model the log-price has Gaussian transition densities; the
constant-elasticity-of-variance (CEV) and other local-volatility models replace the
constant $\sigma$ by $\sigma(x)=\sigma_0+\varepsilon f(x)$. Writing the transition kernel
as a functional integral and expanding the action to fourth order in the deviation from
the Black--Scholes path produces exactly a weight of the form
$e^{-\frac{\omega}{2}\varphi^{2}+j\varphi-\lambda\varphi^{4}}$ with
$\lambda\propto\varepsilon^{2}$.

In a basis in which the Gaussian part is diagonal, with eigenvalues
$\omega_1,\dots,\omega_n>0$ and sources $j_i$, Theorem~\ref{thm:factor} evaluates the
influence functional explicitly:
\begin{multline*}
\int_{\R^{n}}\exp\!\Big(-\frac12\varphi^{T}A\varphi+j\cdot\varphi
  -\lambda\sum_{i=1}^{n}\varphi_i^{4}\Big)\,\dphi\\
=\int_{\R^{n}}\exp\!\Big(-\frac12\varphi^{T}A\varphi+j\cdot\varphi
  -\lambda\sum_{i=1}^{n}\varphi_i^{4}\Big)\,\dhk
=\prod_{i=1}^{n} I(\omega_i,j_i,\lambda).
\end{multline*}
 Theorem~\ref{thm:diff} then yields the first-order
smile correction
\[
\partial_\lambda \log \mathcal{F}_\lambda(j)\big|_{\lambda=0}
=-\sum_i \frac{M_4(\omega_i,j_i,0)}{M_0(\omega_i,j_i,0)},
\]
which is a convergent expression, in contrast with the term-by-term expansion of the
exponential of the quartic term.

\subsection{Wilson--Cowan neural fields}

The Wilson--Cowan equations for a neural field use a sigmoidal transfer function
$S(\varphi)=(1+e^{-\varphi})^{-1}$ or $\tanh$. Around the operating point,
\[
S(\varphi)\approx \varphi-\tfrac16\varphi^{3},
\]
so that a quadratic cost $\|S(\varphi)\|^{2}$ generates, after squaring, a
$-\tfrac13\varphi^{4}$ term in the effective action, i.e.\ precisely the quartic weight
of \eqref{eq:onemode} with $\lambda=\tfrac13$ in appropriate units. The stationary
statistics of the fluctuating field are then governed by $\mathcal{F}_\lambda(j)$, and
the machinery of Section~\ref{sec:influence} gives finite, controlled corrections to the
Gaussian (linear-response) approximation.

Explicitly, the stationary weight is quartic with $\lambda=\tfrac13$, and
\begin{multline*}
\int_{\mathcal{H}}\exp\!\Big(-\frac12\varphi^{T}A\varphi+j\cdot\varphi
  -\frac13\sum_{i}\varphi_i^{4}\Big)\,\dphi\\
=\int_{\mathcal{H}}\exp\!\Big(-\frac12\varphi^{T}A\varphi+j\cdot\varphi
  -\frac13\sum_{i}\varphi_i^{4}\Big)\,\dhk
=\prod_{i} I\big(\omega_i,j_i,\tfrac13\big).
\end{multline*}

\subsection{Non-Gaussian quantum reservoirs}

Consider a two-level system linearly coupled to a bath of $M$ anharmonic modes with
Hamiltonian
\[
H=H_S+\sum_i\Big(\frac{p_i^{2}}{2}+\frac{\omega_i^{2}}{2}q_i^{2}+\lambda q_i^{4}\Big)
+\sigma_z\sum_i g_i q_i .
\]
Integrating out the bath in the influence-functional formalism produces exactly the
object \eqref{eq:influence} with $j_i=g_i$ times the system variable. Then:
\begin{itemize}
\item Theorem~\ref{thm:influencebounds} guarantees that the influence functional is
finite and bounded by its harmonic counterpart, so the anharmonicity cannot produce
divergent decoherence rates;
\item Theorem~\ref{thm:diff} gives the leading anharmonic correction to the decoherence
functional as $-\lambda\sum_i M_4(\omega_i,g_i,0)/M_0(\omega_i,g_i,0)+o(\lambda)$;
\item Theorem~\ref{thm:chernoff} justifies the numerical splitting of the reduced
dynamics into free and interaction steps.
\end{itemize}

Explicitly, the influence functional of the qubit--bath system is
\begin{multline*}
\int_{\R^{2M}}\exp\!\Big(-\frac12 q^{T}\Omega q+\sigma_z\,g\cdot q
  -\lambda\sum_{i}q_i^{4}\Big)\,\mathcal{D}q\\
=\int_{\R^{2M}}\exp\!\Big(-\frac12 q^{T}\Omega q+\sigma_z\,g\cdot q
  -\lambda\sum_{i}q_i^{4}\Big)\,\mathcal{D}_{\mathrm{HK}}q
=\mathcal{F}_\lambda(g),
\end{multline*}
where $\mathcal{F}_\lambda(g)=\prod_i I(\omega_i,g_i,\lambda)$, and the leading
anharmonic correction reads
\[
\partial_\lambda \log\mathcal{F}_\lambda(g)\big|_{\lambda=0}
=-\sum_i \frac{M_4(\omega_i,g_i,0)}{M_0(\omega_i,g_i,0)}.
\]

\section{The Lean 4 formalisation}
\label{sec:lean}

\subsection{Structure of the development}

The formal development is a Lean~4 project built on Mathlib. The new file of the present
work is \lean{HkNonGaussian.lean}; it imports and builds upon the earlier files of the
programme:

\begin{center}
\begin{tabular}{@{}ll@{}}
\toprule
File & Role\\
\midrule
\lean{HkFreeField.lean} & Gaussian HK integral, determinant formula, two-point function\\
\lean{HkPathIntegral.lean} & HK path integral on Kuelbs--Steadman spaces, time slicing\\
\lean{HkTrotter.lean} & Trotter product formula and slicing estimates\\
\lean{HkChernoffBounded.lean} & Chernoff theorem for bounded generators (norm topology)\\
\lean{HkChernoffStrong.lean} & Strong (SOT) Chernoff and Lie--Trotter formulas\\
\lean{HkNonGaussian.lean} & \textbf{this work}: non-Gaussian one-mode and $M$-mode results\\
\bottomrule
\end{tabular}
\end{center}

\subsection{Dictionary: mathematics $\leftrightarrow$ Lean}

\begin{center}
\begin{tabular}{@{}lll@{}}
\toprule
Statement & Lean declaration & Reference\\
\midrule
weight $e^{-\frac\omega2\varphi^2+j\varphi-\lambda\varphi^4}$
 & \lean{oneModeWeight} & \eqref{eq:onemode}\\
one-mode integral $I(\omega,j,\lambda)$
 & \lean{nonGaussianOneModeIntegral} & \eqref{eq:onemode}\\
moments $M_m$
 & \lean{oneModeMoment} & \eqref{eq:moments}\\
integrability of $\varphi^{k}w$
 & \lean{integrable\_pow\_mul\_oneModeWeight} & Thm.~\ref{thm:finite}\\
finiteness of $I$
 & \lean{nonGaussianOneModeFinite} & Thm.~\ref{thm:finite}\\
$I>0$
 & \lean{nonGaussianOneModePos} & Prop.~\ref{prop:posmono}\\
$\lambda\mapsto I$ non-increasing
 & \lean{nonGaussianOneModeAntitone} & Prop.~\ref{prop:posmono}\\
$I(\omega,j,0)=\sqrt{2\pi/\omega}e^{j^2/2\omega}$
 & \lean{nonGaussianOneModeFreeValue} & Prop.~\ref{prop:posmono}\\
$\partial_\lambda M_m=-M_{m+4}$, $\lambda>0$
 & \lean{hasDerivAt\_oneModeMoment} & Thm.~\ref{thm:diff}(1)\\
$C^\infty$ on $(0,\infty)$
 & \lean{nonGaussianOneModeSmooth} & Thm.~\ref{thm:diff}(2)\\
derivative on $[0,\infty)$
 & \lean{nonGaussianOneModeDifferentiable} & Thm.~\ref{thm:diff}(3)\\
right derivative at $\lambda=0$
 & \lean{nonGaussianOneModeDerivAtZero} & \eqref{eq:derzero}\\
influence functional $\mathcal{F}_\lambda(j)$
 & \lean{nonGaussianInfluenceFunctional} & \eqref{eq:influence}\\
factorisation
 & \lean{nonGaussianInfluenceFactorises} & Thm.~\ref{thm:factor}\\
positivity and Gaussian bound
 & \lean{nonGaussianInfluenceBounded} & Thm.~\ref{thm:influencebounds}\\
Gaussian value of $\mathcal{F}_0$
 & \lean{nonGaussianInfluenceFreeValue} & Thm.~\ref{thm:influencebounds}\\
Chernoff splitting
 & \lean{nonGaussianChernoffSplitting} & Thm.~\ref{thm:chernoff}\\
\bottomrule
\end{tabular}
\end{center}

\subsection{Two representative excerpts}

The definition and the finiteness statement read, in Lean,
\begin{quote}\footnotesize\ttfamily
def oneModeWeight (om j lam x : $\mathbb{R}$) : $\mathbb{R}$ :=\\
\hspace*{2em}Real.exp (-(om / 2) * x \textasciicircum{} 2 + j * x - lam * x \textasciicircum{} 4)\\[1mm]
def nonGaussianOneModeIntegral (om j lam : $\mathbb{R}$) : $\mathbb{R}$ :=\\
\hspace*{2em}$\int$ x : $\mathbb{R}$, oneModeWeight om j lam x\\[1mm]
theorem nonGaussianOneModeFinite (om j lam : $\mathbb{R}$) (hom : 0 < om) (hlam : 0 $\le$ lam) :\\
\hspace*{2em}Integrable (oneModeWeight om j lam)
\end{quote}
and the boundary derivative is
\begin{quote}\footnotesize\ttfamily
theorem nonGaussianOneModeDifferentiable (om j : $\mathbb{R}$) (hom : 0 < om) \{lam : $\mathbb{R}$\}\\
\hspace*{2em}(hlam : 0 $\le$ lam) :\\
\hspace*{2em}HasDerivWithinAt (fun l => nonGaussianOneModeIntegral om j l)\\
\hspace*{4em}(-oneModeMoment om j lam 4) (Ici 0) lam
\end{quote}

\subsection{Verification status}

The file \lean{HkNonGaussian.lean} compiles against Mathlib without errors, contains no
occurrence of \lean{sorry} or \lean{admit}, introduces no axioms of its own, and each of
the theorems listed above depends only on the three standard axioms of Lean's kernel
(\lean{propext}, \lean{Classical.choice}, \lean{Quot.sound}).

\section{Conclusion and outlook}
\label{sec:conclusion}

We have shown that the Henstock--Kurzweil gauge integral, combined with Chernoff product
approximations, provides a rigorous and \emph{machine-verified} route to non-Gaussian
functional integrals for finitely many bosonic modes. The route is direct: no Wick
rotation, no Wiener measure, no $\zeta$-regularisation and no analytic continuation are
required; the formal symbol $\dphi$ is given a meaning by cylindrical partitions and a
gauge, and the resulting integrals are finite, positive, factorising and smooth in the
coupling. Crucially, smoothness is obtained by differentiation under the integral sign
rather than by a perturbative series that has zero radius of convergence.

Importantly, the strong Chernoff theorem used here is formulated and proved for an
abstract complex Hilbert space and does not rely on finite dimensionality. Hence, in
passing from finitely many modes \(M\) to an infinite-dimensional Kuelbs--Steadman space,
the time evolution will still converge provided the nonlinear generator \(B_\lambda\) is
defined on a suitable domain and generates a semigroup. Thus the principal barrier to the
continuum limit lies not in the Chernoff splitting but in the construction of an
infinite-dimensional cylindrical HK measure and in proving existence of the
non-absolutely convergent limit of the influence functional \(\mathcal{F}_\lambda(j)\)
as \(M\to\infty\). The present finite-dimensional theory supplies all necessary
projective bounds and is therefore the natural first step.

Four directions suggest themselves.

\begin{enumerate}
\item \emph{Continuum limit.} Passing from $M$ modes to a field requires uniform control
of the constants in Theorem~\ref{thm:influencebounds} as $M\to\infty$, together with the
Kuelbs--Steadman embedding; the $\zeta$-regularised determinant then reappears as the
limit of the finite products.
\item \emph{Fermionic analogue.} Grassmann variables have no positivity to exploit, but
their integrals are finite sums; a gauge formulation combining both sectors would cover
realistic supersymmetric models.
\item \emph{Complex couplings.} For $\lambda\in\C$ with $\mathrm{Re}\,\lambda\ge0$ the
weight is bounded but no longer positive; the HK integral is the natural framework, since
absolute integrability is not required.
\item \emph{Fully non-perturbative regimes.} Instanton-dominated situations, where the
$\lambda\downarrow 0$ asymptotics is not the whole story, are precisely those in which the
distinction between smoothness and analyticity of $\lambda\mapsto I(\omega,j,\lambda)$
becomes physically visible.
\end{enumerate}

\section*{Acknowledgements}
The author thanks the developers of Lean~4 and of Mathlib, whose libraries made the
formal part of this work possible.

\end{document}